\documentclass[acmsmall,screen]{acmart}
\startPage{1}

\setcopyright{none}

\usepackage{booktabs}   
\usepackage{subcaption} 

\usepackage{graphicx}
\graphicspath{ {./images/} }

\usepackage{tikzit}

\tikzstyle{Lang}=[fill={rgb,255: red,66; green,173; blue,255}, draw={rgb,255: red,2; green,145; blue,255}, shape=circle, minimum width=1.2cm, font={\small}]
\tikzstyle{Param}=[fill={rgb,255: red,255; green,231; blue,133}, draw={rgb,255: red,255; green,204; blue,0}, shape=rectangle, fill opacity=.5, rounded corners=0.1cm]
\tikzstyle{Plus}=[fill=white, draw=black, shape=circle]
\tikzstyle{langGroup}=[fill={rgb,255: red,66; green,173; blue,255}, draw={rgb,255: red,2; green,145; blue,255}, shape=circle, fill opacity=.5, rounded corners=0.1cm]
\tikzstyle{Param Box}=[draw={rgb,255: red,255; green,204; blue,0}, fill={rgb,255: red,255; green,231; blue,133}, rounded corners=0.1cm]
\tikzstyle{LangBox}=[fill={rgb,255: red,66; green,173; blue,255}, draw={rgb,255: red,2; green,145; blue,255}, shape=rectangle]
\tikzstyle{Text}=[fill=none, draw=none, shape=circle, font={\small}]

\tikzstyle{dash compiler}=[draw={rgb,255: red,2; green,145; blue,255}, line width=0.05cm, dashed, >=stealth, ->]
\tikzstyle{compiler}=[draw={rgb,255: red,2; green,145; blue,255}, ->, line width=0.05cm, >=stealth]
\tikzstyle{Param Fn}=[draw={rgb,255: red,255; green,204; blue,0}, ->]
\tikzstyle{Dep}=[draw={rgb,255: red,44; green,166; blue,93}, ->, line width=0.05cm, dashed]
\tikzstyle{GreenBox}=[-, fill={rgb,255: red,58; green,217; blue,122}, draw={rgb,255: red,44; green,166; blue,93}, rounded corners=0.1cm, fill opacity=0.5]
\tikzstyle{Green}=[draw={rgb,255: red,44; green,166; blue,93}, ->]

\usepackage[utf8]{inputenc}
\usepackage[T1]{fontenc}
\usepackage{url}

\usepackage{mathpartir}

\usepackage{wrapfig}

\newcommand{\fwk}{Pyrosome}

\newcommand{\tnat}{\text{nat}}
\newcommand{\tunit}{\text{unit}}
\newcommand{\bool}{\text{bool}}
\newcommand{\tru}{\text{true}}
\newcommand{\fls}{\text{false}}
\newcommand{\ite}[3]{\text{if }#1\text{ then }#2\text{ else }#3}
\newcommand{\eset}[2]{\text{set }#1 := #2}
\newcommand{\esetCPS}[2]{\eset {#1} {#2}\text{ in }}
\newcommand{\eget}[1]{\text{get }#1}
\newcommand{\egetCPS}[2]{\eget {#1} \text{ as } #2\text{ in }}

\newcommand{\efix}[4]{\text{fix }#1(#2:#3):=#4}
\newcommand{\efixCC}[0]{\text{fix }}
\newcommand{\elamNoVar}[1]{\lambda#1.~}
\newcommand{\elam}[2]{\elamNoVar{(#1:#2)}}
\newcommand{\eret}[1]{\text{ret }#1}
\newcommand{\epair}[2]{\langle #1, #2 \rangle}
\newcommand{\pmPair}[3]{\text{let } \langle #2, #3 \rangle := #1 \text{ in }}
\newcommand{\eclo}[5]{\text{clo }\langle(#1:#2 \times #3).#4, #5\rangle}
\newcommand{\Ehole}{[~]}
\newcommand{\Eplug}[2]{#1[#2]}

\newcommand{\cmp}[1]{\llbracket #1 \rrbracket}

\newcommand{\presNoArgs}{\textsf{Preserving}}
\newcommand{\pres}[3]{\presNoArgs(#1,#2,#3)}
\newcommand{\presExt}[4]{\presNoArgs_{#1}(#2,#3,#4)}

\newcommand{\bindCPS}[2]{\text{bind } #1 := #2;~}

\newcommand{\unpack}[3]{\text{let}~\langle #1,#2 \rangle := #3 ~\text{in}~}
\newcommand{\pack}[2]{\langle #1,#2 \rangle}

\newcommand{\tvec}{\text{vec~}}
\newcommand{\vnil}{\text{nil}}
\newcommand{\vcons}{\mathbin{::}}
\newcommand{\vapp}[1]{\text{app}~#1~}

\newcommand{\sctx}{\text{ctx}}
\newcommand{\sty}{\text{ty}}
\newcommand{\sexp}[2]{\text{exp}(#1,#2)}
\newcommand{\sval}[2]{\text{val}(#1,#2)}
\newcommand{\ssub}[2]{\text{sub}(#1,#2)}
\newcommand{\scmp}[1]{\text{cmp}(#1)}
\newcommand{\sconfig}[2]{\text{config}(#1,#2)}

\newcommand{\natv}{\text{nv }}

\newcommand{\Iskip}{\text{skip}}
\newcommand{\Iseq}[2]{#1;#2}
\newcommand{\Iassign}[2]{\text{assign }#1 := #2}

\newcommand{\Iwhile}[2]{\text{while }#1~ \{ #2 \}}

\begin{document}

\title{Pyrosome: Verified Compilation for Modular Metatheory}

\author{Dustin Jamner}
\orcid{0000-0003-0700-3514}
\affiliation{%
  \institution{Massachusetts Institute of Technology}
  \city{Cambridge}
  \state{Massachusetts}
  \country{USA}
}
\email{dijamner@mit.edu}

\author{Gabriel Kammer}
\affiliation{%
  \institution{Intel}
  \country{USA}
}
\email{gkammer@mit.edu}

\author{Ritam Nag}
\affiliation{%
  \institution{Massachusetts Institute of Technology}
  \city{Cambridge}
  \state{Massachusetts}
  \country{USA}
}
\email{rnag@mit.edu}

\author{Adam Chlipala}
\affiliation{%
  \institution{Massachusetts Institute of Technology}
  \city{Cambridge}
  \state{Massachusetts}
  \country{USA}
}
\email{adamc@csail.mit.edu}

\begin{abstract}
  We present Pyrosome,\footnote{Pyrosomes are tiny colonial organisms
  that connect to each other to form tube-shaped colonies
  up to 60 feet in length. In our framework, compilers are similarly made up of many small, independent components.} a generic framework for modular language metatheory
  that embodies a novel approach to extensible semantics and compilation, implemented in Coq.
  Common techniques for semantic reasoning are often tied to the specific structures of the languages
  and compilers that they support.
  In Pyrosome, verified compilers are fully extensible,
  meaning that to extend a language (even with a new kind of effect) simply requires defining and verifying the compilation of the new feature,
 reusing the old correctness theorem for all other cases.
  The novel enabling idea is an inductive formulation of equivalence preservation that supports the addition of new rules to the source language, target language, and compiler.

  Pyrosome defines a formal, deeply embedded notion of programming languages
  with semantics given by dependently sorted equational theories,
  so all compiler-correctness proofs boil down to type-checking and equational reasoning.
  We support vertical composition of any compilers expressed in our framework in addition to feature extension.
  As a case study, we present a multipass compiler
  from System F with simple references, through CPS translation and closure conversion.
  Specifically, we demonstrate how we can build such a compiler incrementally
  by starting with a compiler for simply typed lambda-calculus and adding natural numbers, the unit type, recursive functions, and a global heap, then extending judgments with a type environment and adding type abstraction,
  all while reusing the original theorems.
  We also present a linear version of the simply typed CPS pass
  and compile a small imperative language to the simply typed target to show how Pyrosome handles substructural typing and imperative features.
\end{abstract}

\ccsdesc[500]{Software and its engineering~General programming languages}
\ccsdesc[300]{Social and professional topics~History of programming languages}

\maketitle

\section{Introduction}

\label{sec:intro}

Compiler verification is a laborious process compared to unverified implementation.
Nevertheless, numerous researchers have taken on the challenge \cite{compcert,chlipala2010verified,cake}
due to the significant benefits of formal verification \cite{csmith}.
Compilers are one of the most critical classes of programs to which we can apply formal verification
due to their presence in almost every software pipeline.
To verify a full software system requires verifying any compilers involved \cite{lightbulb},
and even in unverified software stacks,
the volume of code that depends on a given compiler often warrants the work of providing strong guarantees.

Despite the clear desirability of formally verified compilers,
the overwhelming majority of compilers in use today are unverified.
Examining the current state of compiler verification,
existing efforts suffer from one or more limitations,
with the end result that a few common requirements of software systems cannot be satisfied by current techniques.
To begin with, existing projects' internal mechanics are typically closely tied to their chosen source languages
and often to their implementations.
As a result, augmenting their results to support more permissive linking and extensibility quickly grows difficult \cite{kang2016,patterson2019next}.
While recent work has made advances in high-to-low-level and multilanguage linking \cite{compcerto,patterson-semantics},
it typically does so by lowering all reasoning to a common target semantics
rather than by connecting source-level reasoning directly.

In general, much of the compilation literature revolves around establishing single, all-encompassing models
that form key dependencies of all steps of the verification process.
This pattern drives the aforementioned limitations since
correct compilers tend to be bespoke artifacts that revolve around their central models.
For example,
\citet{perconti2014verifying} proved a correctness result that supported linking with arbitrary target programs,
but their multilanguage approach builds all languages in the compiler into the statement of the correctness theorem, making it difficult to extend with new features.
This limitation is also the root cause of another key factor in the slow adoption of verified compilation.
Most major programming languages are living projects with routinely evolving specifications.
State-of-the-art verification efforts view source-language specifications as single, monolithic entities.
As a result, the kind of incremental improvements one might expect to see in an actively evolving language
pose a serious threat to attempts at formal verification.
While many such changes might be handled by effective automation and proof engineering,
the cross-cutting nature of specification changes makes it difficult to estimate the cost of updating the verification since there is no hard upper bound on the necessary revisions.

In this work, we present a foundational study demonstrating language-agnostic metatheory and techniques
that combine to enable fully extensible compiler verification, with a case study at the core-calculus level.
While we hope one day to expand our software artifact's scope to cover the full feature set required to implement a practically usable compiler,
our present aim is to lay out the definitions, metatheorems, and proof structures that support translation extensibility in the abstract and demonstrate their behavior on a small research calculus.
We focus on cross-language translation passes in particular because,
while intralanguage optimizations often make up the bulk of practical compilers,
from a specification standpoint, a generic theory of cross-language properties is the greater challenge.

The particular combination of design decisions embodied in Pyrosome includes among other things: choosing an answer to the question of what constitutes a programming language;  demarcating a well-behaved formulation of translation passes; and choosing how to formulate compiler correctness, including eschewing contextual equivalence.
These three decisions come together to enable strong metatheoretical principles that support our desired notion of extensibility.
We represent programming languages as Generalized Algebraic Theories (GATs) \cite{cartmell,QIITs},
which define program semantics via lists of judgments that generate well-formedness predicates and minimal equational theories.\footnote{As opposed to contextual equivalence, which is a maximal equational theory. \cite{plotkin1977lcf}}
The translations we primarily consider are type- and equivalence-preserving maps (GAT morphisms) that are defined pointwise over the constructs of the source language and so are amenable to extension.
Our overall contributions are as follows:
\begin{itemize}
\item We develop metatheory, methods, and mechanized tooling for building and combining language extensions and compilers in the setting of GATs.
\item We define an inductive characterization of well-formedness for cross-language translations encoded in Pyrosome
  and prove that it implies both type and equivalence preservation.
\item We characterize extensions of these specifications and compilers and prove theorems enabling separate
  reasoning about and combination of such extensions.
\item We implement a generic partial-evaluation pass to support the viability of optimizations in this context.
\item We build a multipass compiler for System F with recursive functions, existential types, and a global heap
  that transforms it into CPS and then performs closure conversion,
  by starting from a compiler for STLC and successively adding features via our extension theorems.
\end{itemize}

\section{Generalized Algebraic Theories}
\label{sec:GAT}

Generalized Algebraic Theories (GATs) \cite{cartmell,QIITs} lie at the core of Pyrosome's metatheory.
GATs are dependently typed algebraic theories, originally developed to capture dependent type theories.
Like simple algebraic theories, a GAT consists of a list of judgment rules introducing function symbols and equations.
However, unlike simple theories, in a GAT, each rule can depend not only on the prior function symbols but also on the prior equations to justify its well-formedness, as we will see later.
While recent work has used intrinsically typed structures \cite{QIITs},
in Pyrosome, we define all data (terms, sorts, rules, and languages) as simply typed ASTs
and encode this dependent structure in separate well-formedness predicates over these constructs
as \citet{cartmell} did
in order to keep our computation simply typed at the Coq level.

\begin{figure}
  \begin{mathpar}
    \inferrule{ }{\tnat ~\mathit{sort}} \and
    \inferrule{ }{0 : \tnat} \and
    \inferrule{n : \tnat}{ S n : \tnat} \and
    \inferrule{m : \tnat \and n : \tnat}{ m + n : \tnat} \and
    \inferrule{n : \tnat}{ 0 + n = n : \tnat}\and
    \inferrule{m : \tnat \and n : \tnat}{ S m + n = S (m + n) : \tnat}
  \end{mathpar}
  \caption{GAT for natural numbers}
  \label{fig:nat-GAT}
\end{figure}
We will use a brief example to explain the different forms of rules that make up a GAT.
\autoref{fig:nat-GAT} defines a GAT for natural numbers with addition.
The first rule declares that $\tnat$ is a sort of our theory.
This theory only uses one simple sort, so it could be expressed as an untyped algebraic theory,
but we will expand it shortly.
The next three rules declare $0$, $S$, and $+$ as 0-, 1-, and 2-argument terms of sort $\tnat$ respectively
and specify the sorts of each's arguments above the line.
The last two rules define the behavior of addition via equations on the symbols of the theory.
Equality between terms in a GAT is given by the reflexive, symmetric, transitive, congruence closure of the equations included in the GAT definition.

\begin{figure}
  \begin{mathpar}
    \inferrule{n : \tnat}{\tvec n ~\mathit{sort}} \and
    \inferrule{ }{\vnil : \tvec 0} \and
    \inferrule{m : \tnat \and n : \tnat \and v : \tvec n}{m \vcons v : \tvec S n} \and
    \inferrule{m : \tnat \and v : \tvec m \and n : \tnat \and v' : \tvec n}
              {\vapp v v' : \tvec (m + n)} \and
    \inferrule{n : \tnat \and v : \tvec n}{ \vapp {\vnil} v = v : \tvec n}\and
    \inferrule{i : \tnat \and m : \tnat \and v : \tvec m \and n : \tnat \and v' : \tvec n}
              {\vapp{(i \vcons v)} v' = i \vcons (\vapp v v') : \tvec (S m + n)}
  \end{mathpar}
  \caption{GAT for vectors}
  \label{fig:vec-GAT}
\end{figure}

Next we will demonstrate a GAT that makes use of dependent typing by adding more rules to this theory to develop a datatype of known-length vectors (of nats for simplicity).
As shown in \autoref{fig:vec-GAT}, we start with adding a second sort to our theory, the dependent sort of $n$-length vectors.
Among the subsequently defined vector term formers, $\vnil$, $n\vcons v$, and app, each specifies its length as a component of its sort.
Finally, we bestow app with computational behavior via equations that specify its action on $\vnil$ and $n\vcons v$ on the left.
This last rule demonstrates what we mentioned earlier about rules depending on prior equations.
If we calculate the natural type for the left-hand side of the equation,
we get $\tvec (S m + n)$ as per the annotation.
However, if we do the same for the right-hand side, we get $\tvec S (m + n)$ instead.
The reason we can still use this term at our desired type is by invoking the defining equation of addition from \autoref{fig:nat-GAT}, which the GAT allows implicitly.
While our simply typed case studies do not use this expressive strength,
it is critical to supporting our polymorphic extension.

\section{Simply Typed Lambda Calculus}
\label{sec:STLC}

We begin our case studies with the simply typed lambda calculus.
Some frameworks bake in a notion of object-language variables and substitutions \cite{fiore2022,Allais,Blanchette} so that user languages automatically inherit the behavior of binders,
which makes defining STLC a quick task.
While this convention is suitable for an important class of theories, many languages, for example linear languages or those with dynamic scope, do not fit within these models.
The expressiveness of GATs allows us to define not just the expressions but also the types, contexts, and judgment forms of a programming language as terms and sorts in a GAT,
so the object languages are free to define their own notions of substitution.
On the other hand, since such constructs exist internally to the theories,
GATs do not come equipped with a primitive notion of binding.
Thus, as the cost of this freedom, to define STLC we must first define a call-by-value substitution calculus.
However, once we have done so, we can reuse it as the basis for many languages.

We base the substitution calculus off of \citet{dybjer1995internal}
but modify it to represent call-by-value substitutions,
with a subset of the rules shown in \autoref{fig:subst}.
The first five rules declare sorts for object-language contexts, types, values, expressions, and substitutions.
We elide metavariable sorts of the forms $\Gamma : \sctx$ and $A : \sty$ in \autoref{fig:subst}
when they can be inferred from the sorts of the other metavariables,
although formally every metavariable in a rule has a specified sort.
The syntactic form $\eret v$ injects values into the sort of expressions.
Its name is inspired by fine-grained call-by-value calculi \cite{levy2003modelling},
although our example remains standard call-by-value.
We define an operation to append a type to a context and
write de Bruijn index $\underline{0}$ to represent the variable at the head of a context.
We also add syntactic forms for explicit substitutions \cite{abadi1989explicit} on expressions
and values and a substitution $\langle\gamma,v\rangle$ that appends a value to the head of a substitution.
We include two representative equations in \autoref{fig:subst}.
The first defines the behavior of expression substitution on $\eret v$:
it pushes under the syntactic form to apply a value substitution to the subterm.
The second equation defines part of the behavior of substitutions on de Bruijn indices,
specifically the property that index $\underline{0}$ projects the head value out of any substitution applied to it.
Other rules of the calculus handle features like weakening and substitution composition.

\begin{figure}
  \begin{mathpar}
    \inferrule{ }{\sctx ~\mathit{sort}} \and
    \inferrule{ }{\sty ~\mathit{sort}} \and
    \inferrule{\Gamma : \sctx \and A : \sty}{\sval \Gamma A ~\mathit{sort}} \and
    \inferrule{\Gamma : \sctx \and A : \sty}{\sexp \Gamma A ~\mathit{sort}} \and
    \inferrule{\Gamma : \sctx \and \Delta : \sctx}{\ssub \Gamma \Delta ~\mathit{sort}} \and
    \inferrule{ v : \sval \Gamma A}{\eret v : \sexp \Gamma A} \and
    \inferrule{\Gamma : \sctx \and A : \sty}{\Gamma,A : \sctx} \and
    \inferrule{\Gamma : \sctx \and A : \sty}{\underline{0} : \sval {(\Gamma,A)} A} \and
    \inferrule{e : \sexp \Gamma A \and \gamma : \ssub \Delta \Gamma}
              {e[\gamma] : \sexp \Delta A} \and
    \inferrule{v : \sval \Gamma A \and \gamma : \ssub \Delta \Gamma}
              {v[\gamma]_v : \sval \Delta A} \and
    \inferrule{\gamma : \ssub \Delta \Gamma \and v : \sval \Delta A}
              {\langle\gamma,v\rangle : \ssub \Delta {(\Gamma,A)}}\and
    \inferrule{v : \sval \Gamma A \and \gamma : \ssub \Delta \Gamma}
              {(\eret v)[\gamma] = \eret v[\gamma]_v : \sexp \Delta A} \and
    \inferrule{\gamma : \ssub \Delta \Gamma \and v : \sval \Delta A}
              {\underline{0}[\langle\gamma,v\rangle]_v = v : \sval \Delta A}
  \end{mathpar}
  \caption{CBV substitution calculus (selected rules)}
  \label{fig:subst}
\end{figure}

We can define STLC as an extension of this calculus by adding a few more rules, just like the vector GAT built on the naturals.
Since we will be working on extensions of the substitution calculus for most of our examples,
we will adopt a few fairly standard notations: we write $A : \sty$ as $\vdash A$,
$e : \sexp \Gamma A$ as $\Gamma \vdash e : A$,
$v : \sval \Gamma A$ as $\Gamma \vdash_v v : A$, and
$\gamma : \ssub \Delta \Gamma$ as $\gamma : \Delta \Rightarrow \Gamma$.
As before, we elide type and context sorts when they can be inferred.
Additionally, we will write subsequent rules as if they were based on a named, rather than de Bruijn-indexed, substitution calculus
for readability.
With that preamble, the base rules of STLC are shown in \autoref{fig:STLC}.
We add three new terms to the syntax: the function type $A \to B$, application expressions $e~e'$, and function values $\elam x A e$.
The subsequent equations define $\beta$-reduction and substitution on applications and functions.
We can (and usually should) also include an $\eta$-expansion equation.
However, the benefit of extensibility is that we treat that variation as just another extension of this base calculus.
It is up to the user whether they intend to support $\eta$ laws, and they can be included or omitted as needed.

\begin{figure}
  \begin{mathpar}
    \inferrule{\vdash A \and \vdash B}{\vdash A \to B} \and
    \inferrule{\Gamma \vdash e : A \to B \and \Gamma \vdash e' : A}{\Gamma \vdash e~e' : B} \and
    \inferrule{\Gamma,x:A \vdash e : B}{\Gamma \vdash_v \elam x A e : A \to B} \and
    \inferrule{\Gamma,x:A \vdash e : B \and \Gamma \vdash_v v : A}
              {\Gamma \vdash (\eret \elam x A e)~ (\eret v) = e[v/x] : B} \and
    \inferrule{\Gamma \vdash e : A \to B \and \Gamma \vdash e' : A \and \gamma : \Delta \Rightarrow \Gamma}
              {\Delta \vdash (e~e')[\gamma] = e[\gamma]~e'[\gamma] : B} \and
    \inferrule{\Gamma,x:A \vdash e : B \and \gamma : \Delta \Rightarrow \Gamma}
              {\Delta \vdash_v (\elam x A e)[\gamma] = \elam x A e[\gamma,x/x]: A \to B}
  \end{mathpar}
  \caption{STLC}
  \label{fig:STLC}
\end{figure}

\subsection{Substitution-Equation Generation}
\label{sec:STLC:sub}
While substitution equations like those for applications and functions are numerous,
they follow a prescribed form.
One benefit of using a deeply embedded representation of GATs is that we can do language metaprogramming in Gallina (Coq's dependently typed metalanguage), which we use to generate such equations automatically rather than write them by hand.
In doing so, we follow a substantial body of research.
For example, the Autosubst line of work used first Ltac and then an external code generator to generate substitution functions for user-defined inductive datatypes \cite{schafer2015autosubst,stark2019autosubst}.
We handle the related situation in our Pyrosome case studies by programmatically generating equations that define the behavior of the explicit substitution syntax from language syntax rules.
For example, the final two equations in \autoref{fig:STLC} are actually autogenerated by this procedure, not handwritten.
While the resulting rules form part of the compiler’s trusted specification, they can be audited straightforwardly.
We apply this machinery to all extensions of the substitution calculus
and additionally port it to the corresponding constructs in the linear substitution calculus used in \autoref{sec:case-study:linear}.

\section{Language Extensions}
In \autoref{sec:GAT} and \autoref{sec:STLC}, we informally described the vector GAT as an extension of the naturals GAT and STLC as an extension of substitution by means of adding rules to the base GAT.
Fortunately, a formal accounting is intuitive: Since GATs are determined by lists of rules, we can express the addition of an extension as the concatenation of the new rules onto the old ones.
The key benefit of GATs for our purposes is that all the well-formedness and equivalence judgments of a GAT are preserved by language extension.
We formalize this family of properties in \autoref{thm:lang-mono}, where $L \subseteq L'$ denotes the unordered inclusion of all rules of $L$ in $L'$.
Here and in subsequent metatheory, we write $C \vdash_L Q$ to mean ``$Q$ is derivable from $C$ in language $L$,'' where $C$ is a dependent list of metavariable sort declarations.

\begin{theorem}[Monotonicity under language extension]
  \label{thm:lang-mono}
  Let $L$ and $L'$ be languages such that $L \subseteq L'$.
  \begin{itemize}
    \item If $C \vdash_L s ~\mathit{sort}$, then $C \vdash_{L'} s ~\mathit{sort}$
    \item If $C \vdash_L t : s$, then $C \vdash_{L'} t : s$
    \item If $C \vdash_L t_1 = t_2 : s$, then $C \vdash_{L'} t_1 = t_2 : s$
  \end{itemize}
\end{theorem}
\begin{proof}
  By mutual induction on all GAT judgments, noting that the cases that make use of $L$ only check for inclusion.
\end{proof}

It is important to remember that the order of rules matters in determining whether a language is well-formed because the well-formedness of each rule might depend on the constructors and equations defined in earlier rules, as we saw with the defining equations for the vector append operation relying on the equations for addition to justify their well-formedness.
However, while the order of a language determines whether it is well-formed, it does not affect its semantics, as demonstrated by our use of list inclusion in \autoref{thm:lang-mono}.
In general, a language remains well-formed under any permutation that produces a topological sort of its dependency graph,
so independent extensions can be swapped freely in a judgment via this family of properties.
For example, consider an extension of the substitution calculus that adds Boolean values $\tru$ and $\fls$ as well as the conditional expression $\ite{e}{e_1}{e_2}$.
Both Booleans and STLC are extensions of substitution, but neither extension depends on the other, so they can be reordered arbitrarily.

\subsection{Equational Theories, Contextual Equivalence, and Extensibility}
\label{sec:ext:ctx-eq}

As we have shown, language semantics in Pyrosome are given directly by equational theories
rather than by defining an operational semantics and deriving contextual equivalence from it,
as is often done in the literature on verified and/or secure compilation \cite{ahmed2011equivalence,devriese2016fully}.
This choice is necessary if we want to have \autoref{thm:lang-mono}.
Contextual equivalence is the maximal equivalence that respects an operational semantics \cite{plotkin1977lcf};
it equates every pair of programs it can while still respecting reduction.
Equational theories on the other hand represent smaller equivalences,
namely the smallest congruent equivalences that respect the chosen rules.
In practice, the axioms of a language in Pyrosome are usually just its reduction rules,
possibly with the inclusion of $\eta$ laws.
Thus, equivalence in Pyrosome relates two programs if and only if their behavior must be equivalent in every implementation and semantic model of the language.

To demonstrate the difference, we will consider the following simple example program, which elides ret for convenience of presentation:
\begin{eqnarray*}
  \mathit{callTwice} &:& (\bool\to \bool) \to \bool \to \bool \\
  \mathit{callTwice} &\triangleq& \elam f {\bool\to \bool}\elam x \bool \\
  && (\elam {\_} \bool f~ x)~ (f~ \fls)
\end{eqnarray*}
When passed $f$ and $x$, it first calls $f$ on $\fls$ and then returns $f~x$.
In STLC, $\mathit{callTwice}$ is contextually equivalent to the identity function $\elam f {\bool \to \bool} f$
since STLC is a pure, terminating language.
However, consider how this program behaves once we add the recursive-function extension in \autoref{fig:rec}.
This extension adds a new value, $\efix f x A e$, that represents a recursive function that may refer to itself as $f$ in its body $e$.
It also adds a $\beta$-reduction rule for interacting with the existing function-application operation from STLC.
With this extension, we can apply $\mathit{callTwice}$ to functions like $\efix f x \bool {\ite x{\tru}{f~x}}$,
which is the identity function on the input $\tru$ but diverges on the input $\fls$.
Since $\mathit{callTwice}$ always starts by applying $f$ to $\fls$, it turns this function from one that terminates on one of two inputs
to one that never terminates, so $\mathit{callTwice}$ is no longer contextually equivalent to the identity function.
Since language extension can so definitively break contextual equivalence, the latter is not a suitable property for Pyrosome to consider.
On the other hand, \autoref{thm:lang-mono} tells us that any pair of functions that we can equate in the GAT for STLC
will still be equivalent once we add recursion or indeed any other extension.

\begin{figure}
  \begin{mathpar}
    \inferrule{\Gamma,f : A \to B,x:A \vdash e : B}{\Gamma \vdash_v \efix f x A e : A \to B} \and
    \inferrule{\Gamma,f : A \to B,x:A \vdash e : B \and \Gamma \vdash_v v : A}
              {\Gamma \vdash (\eret \efix f x A e)~\eret v = e[(\efix f x A e)/f,v/x]:B}
  \end{mathpar}
  \caption{Recursive-function extension}
  \label{fig:rec}
\end{figure}

\section{Cross-Language Translation}

Now that we have established how languages are defined and extended, we can proceed to their compilers.
We define our translations as finite maps from source-language sort and term names
to target-language sorts and terms, which form strict morphisms between GATs.
These finite maps are then folded over a source term, looking up each constructor, to translate it.
Then, just as we can extend our languages by appending new rules, we can extend our
compilers by appending new mappings.
This style formalizes patterns that appear in prior work on compositional compiler correctness
\cite{perconti2014verifying,new2016fully,mates2019under}.

This paper primarily addresses cross-language translations in order to focus on challenges particular to cross-language statements of correctness that can be circumvented in the intralanguage case, for example for optimization passes.
Specifically, there is a central tension in the area of compiler correctness between expressiveness and modularity,
which extends to what translations we consider in this work.
In order to fully support both linking and compiler extension,
we require compilers to operate on terms with arbitrary free metavariables.
Furthermore, we want compilation to be invariant under substitution of metavariables,
i.e. $\cmp{\gamma(e)}=\cmp\gamma(\cmp{e})$,
where $\cmp e$ denotes the compilation of $e$ and $\gamma(e)$ denotes the metavariable substitution $\gamma$ applied to the term $e$.\footnote{Unlike object-language substitution, metavariable substitutions $\gamma(e)$ are a Gallina computation on ASTs rather than another syntactic form. }

We choose to guarantee that this equation holds syntactically.
While we could likely work through much of Pyrosome's metatheory using a superficially more flexible semantic equality,
it is our observation that such a change would do little to change the expressiveness of our compilers. For each term former $f$ of arity $n$,
such a compiler would have to define $\cmp{f~x_1...x_n}$, where $x_1...x_n$ are metavariables.
Then even if the compiler were to inspect subterms deeply when compiling a more complex term $\cmp{f~e_1...e_n}$ for subterms $e_1...e_n$, it would have to produce a result semantically equivalent to $\cmp{f~x_1...x_n}[\cmp{e_1}/x_1,...,\cmp{e_n}/x_n]$.
Since we expect optimization to be handled by intralanguage passes rather than our translations anyway, we choose the simpler option of taking the above term to be the definition of compilation.

\subsection{CPS}
\label{sec:CPS}

\begin{figure}
  \begin{mathpar}
    \inferrule{\Gamma : \sctx}{\scmp \Gamma ~\mathit{sort}} \and
    \inferrule{\vdash A}{\vdash \lnot A} \and
    \inferrule{\Gamma \vdash_v v : \lnot A \and \Gamma \vdash_v v' : A}{\Gamma \vdash v~v'} \and
    \inferrule{\Gamma,x:A \vdash e}{\Gamma \vdash_v \elam x A e : \lnot A} \and
    \inferrule{\Gamma,x:A \vdash e \and \Gamma \vdash_v v : A}
              {\Gamma \vdash (\elam x A e)~v = e[v/x]}\and
    \inferrule{\Gamma \vdash v : \lnot A}
              {\Gamma \vdash_v (\elam x A {v~x}) = v : \lnot A} \and
    \inferrule{\vdash A \and \vdash B}{\vdash A \times B} \and
    \inferrule{\Gamma \vdash_v v : A \times B \and \Gamma, x : A, y : B \vdash e}
              {\Gamma \vdash \pmPair v x y e} \and
    \inferrule{\Gamma \vdash_v v_1 : A \and \Gamma \vdash_v v_2 : B}
              {\Gamma \vdash_v (v_1, v_2) : A \times B} \and    
    \inferrule{\Gamma \vdash_v v_1 : A \and \Gamma \vdash_v v_2 : B \and \Gamma, x : A, y : B \vdash e}
              {\Gamma \vdash \pmPair {(v_1, v_2)} x y e = e[v_1/x,v_2/y]} \and
  \end{mathpar}
  
  \[
  \begin{array}{rcl}
    \cmp{\sval \Gamma A} &\triangleq& \sval {\cmp \Gamma} {\cmp A}\\
    \cmp{\sexp \Gamma A} &\triangleq& \scmp {\cmp \Gamma, k : \lnot \cmp A}\\
    \cmp{A \to B} &\triangleq& \lnot {(\cmp A \times \lnot \cmp B)}\\
    \cmp{\eret v} &\triangleq& k~\cmp v\\
    \cmp{\elam x A e} &\triangleq&
    \elam p {\cmp A \times \lnot \cmp B} {\pmPair p x k \cmp e}\\
    \cmp{e~e'} &\triangleq&
      \bindCPS x {\cmp e}
      \bindCPS y {\cmp {e'}} x~\langle y,k\rangle \text{, where } \\
      &&\bindCPS x e e' \triangleq {e}[{\elam x {B} {e'}}/k]
      \text{ given } \Gamma,k:\lnot B \vdash e\\
  \end{array}
  \]
  \caption{Continuation calculus (selected rules) and CPS translation for STLC}
  \label{fig:cps-lang}
\end{figure}

Our case-study compilers closely follow the designs of the first two passes of \citet{morrisett1999system}, albeit adapted to start with a simply typed compiler before adding in polymorphism.
The CPS translation targets a calculus of continuations, shown in the top half of \autoref{fig:cps-lang}.
This calculus extends the values portion of the base substitution calculus.
Since CPS computations do not return, we replace the expression sort, instead extending the value-substitution calculus
with a computation sort $\scmp \Gamma$ that eliminates the return type,
writing $\Gamma \vdash e$ to indicate a well-formed computation $e : \scmp \Gamma$.
We define $\lnot A$ as the type of a nonreturning continuation that accepts input of type $A$,
named to evoke the definition $\lnot A \triangleq A \to \bot$ and the connection between CPS and double negation.
We also include a positive product-types extension in the CPS target.

With the target out of the way, we can define our first translation.
The second half of \autoref{fig:cps-lang} gives the translation from STLC into
our CPS calculus.
Formally, we encode this translation as an association list mapping term and sort formers of STLC to terms and sorts in the continuation calculus
that may have free metavariables corresponding to the context of the rule associated to the term or sort former.
For example, the $\eret v$ case above represents a pair in the association list
containing the symbol $\eret\!$ and the output term $k~\cmp v$ where the $\cmp v$ on the right-hand side
is a metavariable that will be substituted with the result of compiling the subterm in the $v$ position
when we run the compiler.

To bridge the gap between source-language expression judgments
and target-language computations, we translate expression judgments $\Gamma \vdash e : A$ to
computation judgments $\cmp{\Gamma}, k : \lnot \cmp A \vdash \cmp e$, using the variable $k$ for the continuation.
Functions $A \to B$ map to computations $\lnot {(\cmp A \times \lnot \cmp B)}$,
each taking as inputs the compiled argument $\cmp A$ and the continuation $\lnot \cmp B$.
Since we explicitly represent the injection from values to expressions with the syntax $\eret v$ in the source,
the standard CPS translation of values is split in two:
the translation for $\eret v$ is responsible for calling the continuation,
and the translation for function values (and other value forms as we add them) handles mediating between source and target values.

The compiler case for function application debuts a useful macro, bind, that we will employ in most of our CPS-pass extensions as well.
We define $\bindCPS x e e'$ as the expression substitution that passes a function from $x$ to $e'$ as the continuation to $e$, which satisfies the monad laws implied by its name
in the theory of the continuation calculus.
Most importantly, $(\bindCPS x {\cmp{\eret v}} e') = e'[\cmp v/x]$.
Especially when there are multiple binds, as in function application, we find this notation far more readble than nested, inverted substitutions.

\subsection{Semantics Preservation}
\label{sec:sem-pres}

So far we have defined a few languages and a minimal compilation pass, but we have yet to state or prove theorems.
There are many interpretations of compiler correctness in the literature \cite{patterson2019next}.
The original CompCert \cite{compcert} work proved whole-program simulation
of the source language by the target language and used that relation
to show trace refinement.
Simulation of closed programs is inherently vertically compositional,
which allowed CompCert's proof of correctness to be divided cleanly by pass.
However, modern software relies on linking code from multiple sources,
which a correctness property of this form does not cover \cite{ML-verif-comp}.
Recent developments in the CompCert ecosystem such as CompCertO \cite{compcerto}
support linking, and \citet{10.1145/3632914} improve the state of affairs by properly encapsulating intermediate semantics.
However, they each describe interaction through a subsuming semantic model designed with C-like languages in mind.
This choice does reflect modern languages, which often interoperate with each other through a C interface.
However, we view the state of FFIs as the unfortunate result of the current limitations of language specifications and interfaces, which we believe can be improved via the principles of Pyrosome.
To put it another way, we believe that with proper language and compiler modularity, language interoperation can be freed of its dependence on shared models like C.
Other work on verified compiler correctness solves the linking problem for programs
that share a specific low-level target language~\cite{cito}.
It is not clear how to modify such systems to allow modular addition of new target-language features
or features that extend past the design of the original semantic model.

Pyrosome is designed for proving that compilers are type- and equivalence-preserving
with respect to source and target equational theories, which we express formally in the following definition:
\begin{definition}[Semantics Preservation]
  \label{def:sem-pres}
  A function $\cmp -$ is a \emph{semantics-preserving} compiler from source $S$ to target $T$
  if all of the following hold\footnote{An additional property for preservation of sort equalities is actually included in our artifact, but we do not include it here since we have yet to make use of sort equations, and so it is entailed by these three.}:
  \begin{itemize}
    \item If $C \vdash_S s ~\mathit{sort}$ then $\cmp C \vdash_{T} \cmp s ~\mathit{sort}$
    \item If $C \vdash_S t : s$ then $\cmp C \vdash_{T} \cmp t : \cmp s$
    \item If $C \vdash_S t_1 = t_2 : s$ then $\cmp C \vdash_{T} \cmp{t_1} = \cmp{t_2} : \cmp s$
  \end{itemize}
\end{definition}

We claim that these properties are the critical ones for notions of compiler correctness.
While cross-language relations are more common in the literature,
they typically depend on language-specific relations between observable values or states,
which is less than ideal for generic metatheory.
Furthermore, they can be recovered quickly from semantics preservation when desired.
It turns out that cross-language specifications can readily be reduced to equivalence preservation and a standard inversion lemma.

For example, consider extending the CPS translation to compile Booleans in the source to naturals in the target.
We add the rules $\cmp{\tru} = \text{nv}~ 1$ and  $\cmp{\fls} = \text{nv}~ 0$ where nv injects a separate sort of natural numbers into the sort of values at type $\tnat$.
Note that there are two reasonable cross-language value relations, which we will write as $v_s \sim v_t$,
that one might choose as the specification for a compiler with these rules.
Both relations agree that $\fls \sim \text{nv}~ 0$,
but it would be reasonable either to decide $\tru \sim \text{nv}~ 1$ and no other naturals are allowed,
or to add $\tru \sim \text{nv}~ n$ where $n \neq 0$.
For this example, we will arbitrarily choose the first specification.
Such decisions form an inherent, language-specific component of cross-language compiler specifications and so are unsuitable for, and do not benefit from, inclusion in agnostic scaffolding like Pyrosome.

For now, assume that the CPS translation is semantics-preserving.
The components necessary to bridge to a cross-language relation are inversion and relatedness of values,
precisely the language-specific parts that do not benefit from an agnostic framework.
Additionally, they only deal with the basic properties of values, so lifting to expressions is entirely handled by equivalence preservation.

\begin{theorem}[CPS Cross-Language Correctness]
  \label{thm:cps-cross}
  If $\vdash e = \eret b : \bool$ in STLC + Bool, then $k : \lnot \bool \vdash \cmp e = k~n$ in CPS + Nat
  and $b \sim n$.
\end{theorem}
\begin{proof}
  By equivalence preservation, $k : \lnot \bool \vdash \cmp e = \cmp{\eret b}$.
  By inversion, since $b$ is closed, we have two cases:
  $\vdash_v b = \tru : \bool$ and $\vdash_v b = \fls : \bool$.
  In the first case, $\cmp{\eret \tru}$ evaluates to $k~(\text{nv}~ 1)$ and $\tru \sim \text{nv}~ 1$, so we conclude.
  The second proceeds analogously.
\end{proof}

\subsection{The $\presNoArgs$ Predicate}

Now that we have established semantics preservation as our goal, it remains to establish a standard method for proving it
and for doing so extensibly.
A significant contribution of this work is our definition of a predicate over compilers that follows the structure of the source language and a proof that this predicate implies semantics preservation.
To demonstrate why this predicate is essential to our extensibility results,
consider a standard proof of equivalence preservation for our CPS pass for base STLC.
We assume type preservation in this example for simplicity and walk through the term-equivalence case.
Given two STLC expressions $e_1$ and $e_2$ such that $\Gamma \vdash e_1 = e_2 : A$,
we have to show that $\cmp\Gamma, k : \lnot\cmp A \vdash \cmp{e_1} = \cmp{e_2}$ in the continuation calculus.
We proceed by induction on the proof of $\Gamma \vdash e_1 = e_2 : A$,
which requires us to consider reflexivity, transitivity, symmetry, congruence, and $\beta$-reduction.
The first three hold either trivially or by the inductive hypothesis.
Congruence requires that the compiler be a homomorphism with respect to substitution,
which we have by the structure of translations in Pyrosome.
What remains to prove is the $\beta$-reduction case, where we must prove that the right- and left-hand sides
compile to equivalent terms, although we omit the proof here.

If we examine the structure of this proof sketch,
the only case that depends on the rules of STLC or the definition of the CPS pass is $\beta$-reduction.
The rest can be proven using invariants common across all Pyrosome languages and compilers.
Since they are independent of the compiler under consideration, we can prove them generically,
so we focus on equations from the source-language definition like $\beta$-reduction.
Consider how the proof must be extended to add a new feature, for example product types.
Intuitively, the proof case for $\beta$-reduction should remain the same,
and we must add new cases for the first and second projections out of a pair.
Such an extension is logically straightforward, so we would expect it to be reasonable to formulate mechanically.

However, the structure of our theorem statement, that for all STLC terms,
or for all terms made of product constructs,
compilation preserves equivalence,
is ill-suited to such extension
since it tells us nothing about terms with a mix of product operations and STLC constructs.
To properly reuse our proof about STLC, we need to encapsulate it in a lemma
that can be extended inside the induction on equivalence proofs.

To describe the proof obligations associated with a given source language and compiler,
we define the inductive predicate $\pres{L_t}{cmp}{L_s}$,
which takes as input a target language $L_t$, compiler $cmp$, and source language $L_s$.
This predicate has one constructor for each kind of rule that can be appended to $L_s$.
Each constructor requires that the compiler satisfies $\presNoArgs$ for the tail of $L_s$
plus the appropriate condition for the head depending on the kind of rule:
\begin{itemize}
\item For a sort constructor, the compiler must map it to a well-formed sort in the target.
\item For a term constructor, the compiler must map it to a well-formed term in the target.
\item For an equation, the compiler must map the left- and right-hand sides to equivalent terms in the target.
\end{itemize}

\begin{figure}
  \footnotesize \begin{verbatim}
Context (cmp_pre : compiler).
Inductive preserving_compiler_ext : compiler -> lang -> Prop :=
  | preserving_compiler_nil : preserving_compiler_ext [] []
  ...
  | preserving_compiler_term : forall cmp l n c args e t,
      preserving_compiler_ext cmp l ->
      Model.wf_term (compile_ctx (cmp ++ cmp_pre) c) e (compile_sort (cmp ++ cmp_pre) t) ->
      preserving_compiler_ext ((n, term_case (map fst c) e) :: cmp)
                              ((n, term_rule c args t) :: l)
  | preserving_compiler_term_eq : forall cmp l n c e1 e2 t,
      preserving_compiler_ext cmp l ->
      Model.eq_term (compile_ctx (cmp ++ cmp_pre) c)
              (compile_sort (cmp ++ cmp_pre) t)
              (compile (cmp ++ cmp_pre) e1)
              (compile (cmp ++ cmp_pre) e2) ->
      preserving_compiler_ext cmp ((n, term_eq_rule c e1 e2 t) :: l).
\end{verbatim}
\caption{Coq definition of $\presNoArgs$ (selected cases)}
\label{fig:pres}
\end{figure}

We present the formal definition of $\presNoArgs$ in \autoref{fig:pres},
with the sort cases elided for space as they mirror the term cases.
Ignoring $\texttt{cmp\_pre}$ for the moment, the definition closely matches the prose description.
The empty compiler translates the empty language.
When the top rule of a language defines a new term former named $\texttt{n}$, 
the compiler must map $\texttt{n}$ to a term $\texttt{e}$
such that $\texttt{e}$ is well-formed at a sort and context determined by applying the compiler to the sort and context of the rule.
The rest of the compiler and rest of the language must then satisfy $\presNoArgs$.
In the case of an equality rule, there is no new case of the compiler.
However, there is still a proof obligation: the user must show that the terms on the left and right of the equality,
when compiled to the target language, remain equal.
One significant benefit of working with terms with free metavariables is that this obligation is not quantified at the Coq level, so proving it just requires a single target-language derivation about concrete programs.

The formal definition above features an additional parameter $\texttt{cmp\_pre}$ because it is actually a generalization of $\presNoArgs$ from a whole-pass predicate to a pass-extension predicate.
Let $cmp_{pre}$ be a compiler from source language $L_{pre}$ to target $L_t$,
and let $L_s$ be an extension of $L_{pre}$.
We will write $\presExt{cmp_{pre}}{L_t}{cmp}{L_s}$ for the generalized predicate,
which states that ${cmp}$ is a semantics-preserving compiler extension of $cmp_{pre}$
supporting source extension $L_s$.
We can then define $\pres{L_t}{cmp}{L_s}$ as $\presExt{[]}{L_t}{cmp}{L_s}$, the case where the base language is empty.

The key to proving this property is that each case of $\presNoArgs$
only relies on earlier cases' language rules and mappings in the compiler.
Since the sorts, terms, and contexts that make up each rule of a language
can only reference constructs from earlier rules (those in the tail of the list),
it is sufficient at each rule to be able to compile only the previous constructors.
For example, when compiling STLC using the CPS compiler,
the proof obligation for lambdas references a prefix of the compiler that does not include application,
since the application rule comes later in the language specification.
This ordering is also essential to the modularity of $\presNoArgs$ since it conversely means that proofs of earlier obligations
remain valid as the compiler is extended.

\begin{theorem}[$\presNoArgs$ implies semantics preservation]
  \label{thm:ind-impl-sem}
  Let $L_s$ and $L_t$ be well-formed languages, and let $cmp$ be a compiler.
  If $\pres{L_t}{cmp}{L_s}$, then $cmp$ is a semantics-preserving compiler
  from $L_s$ to $L_t$.
\end{theorem}

For $\presNoArgs$ to be useful, it must give us semantics preservation.
We bridge this gap in \autoref{thm:ind-impl-sem},
which states that the predicate described above implies the universally quantified semantic properties.
The proof of this theorem is by mutual induction over all of the judgment forms in Pyrosome,
where for each judgment, we must show that it is preserved by compilation.
The term-equivalence case resembles the earlier proof we sketched for the CPS pass starting from STLC,
except that we generalize the $\beta$-reduction case to a lemma about the preservation
of each equivalence written in the language description.
This lemma relies on 2 related principles: weakening and monotonicity.
We disallow compilers from overwriting old cases, so we can safely use weakening lemmas
to extend the compilers in the hypotheses provided by $\presNoArgs$
so that they refer to the whole compiler.
To finish lifting each obligation to cover the whole compiler and source language,
we appeal to \autoref{thm:lang-mono}, which states that all judgments are monotonic under language extension.

Thanks to this theorem, we prove semantic preservation for each translation in Pyrosome
by way of $\presNoArgs$.
Our mechanization automatically breaks down the necessary proof obligations,
and the resulting goals are quite amenable to both automation and human reasoning
since they feature no quantification at the Coq level
and can be proven by direct construction of either well-formedness or equivalence derivations in the target language.
In fact, we are now equipped to prove equivalence preservation for our base CPS pass,
which we will do in \autoref{thm:cps-subst-pres} and \autoref{thm:cps-pres},
which state respectively that the translation for the base substitution calculus
and the STLC extension satisfy $\presNoArgs$.

\begin{theorem}[CPS$_{\mathit{subst}}$ is $\presNoArgs$]
  \label{thm:cps-subst-pres}
  Let $\mathit{subst}$ be the source substitution calculus,
  let $\mathit{cps}_{\mathit{subst}}$ be the portion of the CPS pass with $\mathit{subst}$ as its domain,
  and let $\mathit{cont}$ be the continuation calculus.
  Then we have $\presExt{[]}{\mathit{cont}}{cps_{\mathit{subst}}}{\mathit{subst}}$.
\end{theorem}
\begin{proof}
  By the definition of $\presNoArgs$ and construction of $\mathit{cont}$ derivations in all cases.
\end{proof}

\begin{theorem}[CPS is $\presNoArgs$]
  \label{thm:cps-pres}
  Let $cps_{\mathit{subst}}$ be the portion of the CPS pass with $\mathit{subst}$ as its domain,
  let $\mathit{cps}$ be the portion with the STLC extension as its domain,
  and let $\mathit{cont}$ be the continuation calculus.
  Then we have $\presExt{\mathit{cps}_{\mathit{subst}}}{\mathit{cont}}{\mathit{cps}}{\mathit{STLC}}$.
\end{theorem}
\begin{proof}
  By the definition of $\presNoArgs$ and construction of $\mathit{cont}$ derivations in all cases.
   We show the $\beta$-reduction case here as an example (although all cases are fully automated in the mechanization).
  The proof obligation that $\presNoArgs$ generates corresponding to this rule
  requires us to show $\cmp{(\eret \elam x A e)~ \eret v} = \cmp{e[v/x]}$
  in the continuation calculus' equivalence relation.
  By evaluating the compiler and rewriting the term via target-language rules,
  we do so as follows\footnote{We skip the details of stepping through the interaction of the bind macro with the compilation of a return since it involves verbose juggling of continuation substitutions and little else.}:
  \[\begin{array}{rl}
  &\cmp{(\eret \elam x A e)~ (\eret v)}\\
  =& \bindCPS x {\cmp{\eret \elam x A e}} \bindCPS y {\cmp{\eret v}} x~ \langle y, k \rangle\\
  =& \bindCPS x {\cmp{\eret \elam x A e}} (\elam y {\cmp A} {x~ \langle y, k \rangle})~\cmp v \\
  =& \bindCPS x {\cmp{\eret \elam x A e}} {x~ \langle \cmp v, k \rangle} \\
  =& (\elam x {\cmp{A \to B}} {x~ \langle \cmp v, k \rangle})~\cmp{\elam x A e}\\
  =& \cmp{\elam x A e}~ \langle \cmp v, k \rangle\\
  =& (\elam p {\cmp A \times \lnot{\cmp B}}
  {\pmPair p x k {\cmp e}})
  ~ \langle \cmp v, k \rangle\\
  =& \pmPair {\langle \cmp v, k \rangle} x k {\cmp e} \\
  =&\cmp{e}[\cmp v/x, k/k]\\
  =&\cmp{e}[\cmp v/x]\\
  =&\cmp{e[v/x]}\\
  \end{array}\]
\end{proof}

We are almost ready to prove semantics preservation from these facts,
but first we need one more generic theorem:

\begin{theorem}[$\presNoArgs$ Concatenation]
  \label{thm:pres-concat}
  Let $L_s$ and $L_t$ be well-formed languages,
  let $L'_s$ be a well-formed extension of $L_s$,
  and let $cmp$ and $cmp'$ be compilers such that
  $\pres{L_t}{cmp}{L_s}$ and $\presExt{cmp}{L_t}{cmp'}{L'_s}$.
  Then $\pres{L_t}{cmp'+cmp}{L'_s+L_s}$.
\end{theorem}
\begin{proof}
By induction on the proof of $\presExt{L_s}{L_t}{cmp'}{L'_s}$.
\end{proof}

Now we can prove that our CPS pass is semantics-preserving in \autoref{thm:cps-sem-pres}.
We will generally not go through these steps for future case-study components,
since composing the $\presNoArgs$ theorems is entirely mechanical
and indeed performed by a tactic in the mechanization.

\begin{theorem}[CPS is Semantics-Preserving]
  \label{thm:cps-sem-pres}
  Let $\mathit{subst}$ be the source substitution calculus,
  let $\mathit{cps}_{\mathit{subst}}$ be the portion of the CPS pass with $\mathit{subst}$ as its domain,
  let $\mathit{cps}$ be the portion with the STLC extension as its domain,
  and let $\mathit{cont}$ be the continuation calculus.
  Then we have $\mathit{cps} + \mathit{cps}_{\mathit{subst}}$ is a semantics-preserving compiler from
  STLC + $\mathit{subst}$ to $\mathit{cont}$.
\end{theorem}
\begin{proof}
  By \autoref{thm:ind-impl-sem} it suffices to show
  $\pres{\mathit{cont}}{\mathit{cps} + \mathit{cps}_{\mathit{subst}}}{\mathit{STLC} + \mathit{subst}}$.
  We apply \autoref{thm:pres-concat}
  and conclude with \autoref{thm:cps-subst-pres} and \autoref{thm:cps-pres}.
\end{proof}

\subsection{Closure Conversion and Vertical Composition}

\begin{figure}
  \begin{mathpar}
    \inferrule{\vdash A}{\vdash \lnot A} \and
    \inferrule{\Gamma \vdash_v v : \lnot A \and \Gamma \vdash_v v' : A}{\Gamma \vdash v~v'} \and
    \inferrule{z: A \times B \vdash e \and \Gamma \vdash_v v : A}
              {\Gamma \vdash_v \eclo z A B e v : \lnot B} \and
    \inferrule{z: A \times B \vdash e \and \Gamma \vdash_v v : A \and \Gamma \vdash_v v' : B}
              {\Gamma \vdash \eclo z A B e v ~ v'
                = e[\epair v {v'}/z]} \and
    \inferrule{z: A \vdash_v v : \lnot B}
              {z : A \vdash_v \eclo z A B {v[z.1/z]~z.2} z
                = v : \lnot B} \and
  \end{mathpar}
  
  \[
  \begin{array}{rcl}
    \cmp{\cdot} &\triangleq& \tunit\\
    \cmp{\Gamma,x:A} &\triangleq& \cmp\Gamma \times \cmp A\\
    \cmp{\lnot A} &\triangleq& \lnot \cmp A\\
    \cmp{\elam x A e} &\triangleq&
    \text{clo } \langle(x: \cmp\Gamma \times \cmp A). \cmp e, z\rangle\\
    \cmp{v~v'} &\triangleq& \cmp v ~ \cmp{v'}\\
  \end{array}
  \]
  \caption{Closure conversion (selected rules)}
  \label{fig:cc}
\end{figure}

For the second pass of our case study, we perform closure conversion,
materializing the environment as a tuple value.
In \autoref{fig:cc}, we show the closure calculus that replaces our CPS calculus.
To construct environment tuples in a compositional way,
we compile CPS environments to closure-converted types and map CPS
judgments $\Gamma \vdash e$ and $\Gamma \vdash_v v : A$
to closure-converted judgments $z:\cmp\Gamma \vdash \cmp e$ and $z:\cmp\Gamma \vdash_v \cmp v : \cmp A$,
using the variable $z$ for the environment.
In order to define closure conversion in the simply typed setting, we use a fused closure form
$\eclo z A B e v$ that captures the combined behavior of the existential, pair, and function
in the normal translation, in a style closely related to that of \citet{minamide1996typed}.
We define two equations on closures, a $\beta$ rule and an $\eta$ rule.
The $\beta$ rule evaluates the body of the closure, passing in a tuple formed of its argument and its environment.
The $\eta$ rule states that if we have a closure $v$ with a free variable $z$,
then it is equivalent to a new closure that stores $z$ as its environment and calls $v$ in its body.

Now that we have two compiler passes, we can sequence them using function composition
to get from the source to the final target.
Fortunately, we can do so in one step thanks to our choice of definitions.
We prove semantics preservation for closure conversion on its own analogously to CPS.
Then, by transitivity (\autoref{thm:sem-pres-trans}), since both passes preserve semantics, so does their composition.

\begin{theorem}[Transitivity of Semantics Preservation]
  \label{thm:sem-pres-trans}
  If $f$ is a semantics-preserving compiler from $L_1$ to $L_2$
  and $g$ is a semantics-preserving compiler from $L_2$ to $L_3$,
  then $g \circ f$ is a semantics-preserving compiler from $L_1$ to $L_3$.
\end{theorem}
\begin{proof} We show that $g \circ f$ maps equivalent $L_1$ terms to equivalent $L_3$ terms.
  The other conjuncts of semantics preservation proceed analogously.
  Consider $C \vdash t_1 = t_2 : s$ in $L_1$.
  By semantics preservation of $f$, $f(C) \vdash f(t_1) = f(t_2) : f(s)$ in $L_2$.
  Then, by semantics preservation of $g$, $g(f(C)) \vdash g(f(t_1)) = g(f(t_2)) : g(f(s))$ in $L_3$,
  so we conclude.
\end{proof}

\subsection{Recursive-Functions Compiler Extension}
We extended STLC with rules for recursive functions in \autoref{fig:rec}.
To perform CPS translation on recursive functions, we add an analogous
construct for recursive continuations to the continuation calculus
and translate from the first set of rules to the second:
\begin{mathpar}
  \inferrule{\Gamma,f : \lnot A,x:A \vdash e}{\Gamma \vdash_v \efix f x A e : \lnot A} \and
  \inferrule{\Gamma,f : \lnot A,x:A \vdash e  \and \Gamma \vdash_v v : A}
            {\Gamma \vdash (\efix f x A e)~v = e[(\efix f x A e)/f,v/x]}\and
  \cmp {\efix f x A e} \triangleq \efix f z {\cmp A \times \lnot \cmp B} {\pmPair z x k {\cmp e}}
\end{mathpar}

Building on the result for the core CPS transformation,
we show in \autoref{thm:cps-rec} that the compiler for the recursion extension satisfies $\presNoArgs$.
This proof requires fulfilling three new obligations as per the definition of $\presNoArgs$:
one to show that the compiler is type-preserving on recursive functions, one to show that reduction of recursive functions is preserved, and one for the substitution rule for recursive functions, which was autogenerated as discussed in \autoref{sec:STLC:sub}.

\begin{theorem}[CPS$_{rec}$ is $\presNoArgs$]
  \label{thm:cps-rec}
  Let $\mathit{rec}$ be the recursive functions extension to STLC,
  let $\mathit{cps}$ be the CPS pass for STLC,
  let $\mathit{cps}_{\mathit{rec}}$ be the CPS extension with $\mathit{rec}$ as its domain,
  and let $\mathit{cont}_{\mathit{rec}}$ be the continuation calculus extended with recursive continuations.
  Then we have $\presExt{\mathit{cps}}{\mathit{cont}_{\mathit{rec}}}{\mathit{cps}_{\mathit{rec}}}{\mathit{rec}}$.
\end{theorem}
\begin{proof}
  By the definition of $\presNoArgs$ and construction of $\mathit{cont}$ derivations in all cases.
\end{proof}

Similarly to \autoref{thm:cps-sem-pres}, we can show that the combined compiler $\mathit{cps}_{\mathit{rec}} + \mathit{cps}$ is semantics-preserving
by using \autoref{thm:ind-impl-sem} and \autoref{thm:pres-concat}.
However, we also need one more property.
We can only concatenate compilers that have the same target language,
but the codomain of $\mathit{cps}$ is the continuation calculus before we have added recursive continuations.
Fortunately, \autoref{thm:cmp-embed} states that we can lift any compiler targeting a language $L_t$
to one targeting a superset language $L_t'$.
Thus, we can lift $\mathit{cps}$ to target $\mathit{cont}_{\mathit{rec}}$ before appending the $\mathit{cps}_{\mathit{rec}}$ compiler extension.

\begin{theorem}[Compiler codomain embedding]
  \label{thm:cmp-embed}
  If we know $\pres{L_t}{\mathit{cmp}}{L_s}$ and $L_t \subseteq L_t'$,
  then it follows that $\pres{L_t'}{\mathit{cmp}}{L_s}$.
\end{theorem}

Closure conversion for recursive functions requires a bit more care.
The interaction between recursive functions as formulated in our first two calculi and the environment tuple generated by closure conversion is a bit complex when all of the behavior is fused into a single construct, so we separate out a fixpoint combinator in our closure-conversion language:
\begin{mathpar}
  \inferrule{\Gamma\vdash_v v :\lnot(\lnot A \times A) }{\Gamma \vdash_v \efixCC v : \lnot A} \and
  \inferrule{\Gamma\vdash_v v :\lnot(\lnot A \times A)  \and \Gamma \vdash_v v' : A}
            {\Gamma \vdash (\efixCC v)~v' = v~\epair{\efixCC v}{v'}}\and
            \cmp {\efix f x A e} \triangleq \efixCC{\eclo z {\cmp \Gamma} {\lnot(\lnot A \times A)}
                            {\cmp e [\epair{\epair{z.1}{z.2.1}}{z.2.2}/z]} z}            
\end{mathpar}
The translation features some verbose tuple rearrangement,
but in essence it splits the recursive function into the fixpoint combinator
and its argument, which is simultaneously closure-converted.
The separation of the two parts means that we can apply closure laws as normal to reason about the body
and cordon off the recursion so that the two concerns do not interfere with each other.
Just like with the prior passes and extensions, the definition is the greatest difficulty.
The three connected proof obligations are solved automatically in the mechanization other than a single $\eta$ expansion in one case.
From there, we can use our connective theorems as before to join up all relevant components.

Now that we have covered all of the ways compilers can be extended,
we include \autoref{fig:recursive-overview} as a reference for all of the different ways we have extended our compiler.
Solid arrows indicate translations written and verified by the user, and dashed arrows indicate translations
validated by the theorems of Pyrosome that we have introduced over the course of this section.

\begin{figure}
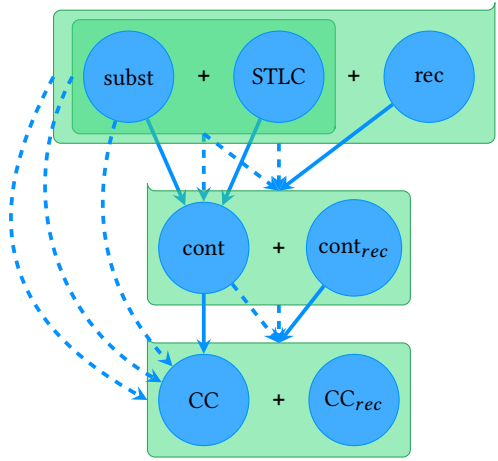

  \ctikzfig{overview}
  \vspace{-8pt}
\caption{Compiler extension}
\label{fig:recursive-overview}
\end{figure}

\subsection{How Incorrect Proofs Fail}

Pyrosome avoids the failure modes typical of trying to extend most prior verified compilers,
especially those with theorems based on contextual equivalence,
by guaranteeing unconditional monotonicity of terms', languages', and compilers' key properties under language extension.
For compiler developers used to contextual equivalence,
this principle can cause some confusion since they are understandably curious about what does happen in situations where
they would expect a new extension to break old equivalences.

Such situations fall into two categories.
The first category covers examples like $\mathit{callTwice}$ in \autoref{sec:ext:ctx-eq}
where the new language feature contradicts a property implicitly assumed by contextual equivalence.
Since semantics preservation in Pyrosome only preserves the least congruence over the explictly given rules,
we avoid this issue by not promising these implicit properties in the first place.
This behavior matches the intent of language specifications in Pyrosome as representing open sets of promises
rather than exhaustive, closed universes.

The second category covers situations where two explicit equations intuitively conflict.
Consider a call-by-value STLC with evaluation contexts that is set up to evaluate the function first in an application;
that is, it has evaluation contexts $E~e$ and $v~E$.
Now say we have two extensions: one that supports mutation and one that adds an additional evaluation context $e~E$,
indicating that evaluation may occur on either the right- or left-hand side first.
It is possible to compile the source language coherently given either one of these extensions.
However, the two of them cannot work together.
What goes wrong concretely in Pyrosome is that the common translation for the core language will validate the equations for at most one of the extensions.
If the compilation for mutation can be proven correct, there will be no reasonable way to satisfy the demands of the evaluation-context exception and vice versa.
Concretely, the obligation generated by $\presNoArgs$ for the failing one will be unprovable due to requiring that two unrelated target programs be shown equivalent.
There is a single unreasonable compiler that satisfies semantics preservation by translating all programs to unit.
However, a trivial check that, for example, the compiler distinguishes 0 and 1, or indeed any cross-language corollary like \autoref{thm:cps-cross},
is adequate to rule out this case.

\subsection{Axiomatic Equivalence and Operational Semantics}

For a language that already has a canonical operational semantics,
we can bridge the gap to equational theories and still derive correctness theorems in terms of these semantics
from Pyrosome's semantics-preservation results.
Generally, each step in the operational semantics is validated by a specific axiom in the equational theory,
which means that reduction immediately implies equivalence.
As a result, if $e \rightarrow^* v$ in the source language,
then $\cmp e = \cmp v$ in the target where $\cmp e$ is the compilation of $e$.
There are then two ways to proceed.
To conclude that the output of a compiler satisfies target-language contextual equivalence, it suffices to show that each rule of the target language's equational theory is validated by the contextual equivalence of its external definition, for example by showing that each equation holds within a logical relation or an equivalence-reflecting denotation.
For a simpler, if weaker, result, it suffices to prove that as with the source language,
each step in the target operational semantics is validated by an equivalence.
We can then use the following theorem to link the two:

\begin{theorem}
  \label{thm:op-bridge}
  Let $\cmp-$ be a semantics-preserving compiler from $S$ to $T$.
  Let $\rightarrow_S$ and $\rightarrow_T$ be subsets of the equivalence relations for $S$ and $T$ respectively.
  Let $V_s$ and $V_t$ be predicates on terms of languages $S$ and $T$
  denoting the observable values of those languages such that if two observable values are equivalent,
  then they are syntactically equal.
  Let $\sim$ relate terms of language $S$ to terms of language $T$ such that for any observable value $v_S$ of $S$,
  there exists an observable value $v_T$ of $T$ such that $\vdash_T \cmp{v_S} = v_T$ and $v_S \sim v_T$.
  
  Then, if $\vdash_S e : t$, $e \rightarrow_S^* v$, and $\cmp e \rightarrow_T^* v'$
  where $v$ and $v'$ are observable values,
  then $v \sim v'$.
\end{theorem}

While the statement of \autoref{thm:op-bridge} is a bit long due to its generality,
its assumptions boil down to requiring one fact each about source and target: that the operational steps are valid equations,
and only one additional fact about the compiler: that the chosen cross-language relation
relates every source value to a value equivalent to its compilation.
As discussed in \autoref{sec:sem-pres}, this latter property can often be validated by simply running the compiler on the appropriate inputs.

\section{Optimization}
\label{sec:optimization}
In our formulation, translation passes are designed for cross-language transformations rather than optimization.
We focus on the cross-language case since it poses the greatest challenges with respect to formulating specifications,
one of the key contributions of this work.
The difference is that optimization passes typically operate within a single language,
so they exist in a much broader design space since source and target code can interact.
Specifically, an optimization pass $o$ can choose to satisfy the specification $C \vdash t = o(t) : s$ for all terms $C \vdash t : s$,
which implies semantics preservation.

\citet{benton2004simple} showed that core optimizations like dead-code elimination and constant propagation can be expressed in an equational theory.
While he justifies such theories by proving them sound with respect to a relational semantics,
we can skip this step and consider the theories definitional, especially since intermediate languages,
where optimizations typically occur,
will disappear from the statement of end-to-end semantics preservation.
More recent work has used both e-graphs and more traditional tree traversals to develop language-generic, rewrite-based optimizations and analyses including uncurrying, dead-code elimination, and linear-algebra optimization \cite{john-li-rewrites,egg,egglog}.
Due to their rewrite-based interfaces, we expect that such tools can be adapted to our domain of equational theory-based language extensibility.

As an initial proof of concept, we implement a language-agnostic partial-evaluation pass that shares its implementation with some of our proof automation.
The partial evaluator takes as input the deeply embedded representation of its host language
and treats its equations as directed rewrite rules, internally producing and checking a proof of equivalence
between its input and output.
By filtering the input language, we can also restrict the pass to, for example, only rules that do not duplicate subterms.
Thus, it is generic over all Pyrosome languages,
can be freely inserted at any stage of a compilation pipeline,
and permits language extension.
While our partial evaluator is fairly na\"ive,
we believe that it shows Pyrosome's potential compatibility with more advanced rewrite-based optimizations.

\section{Global State and Evaluation Contexts}

\begin{figure}
  \begin{mathpar}
    \inferrule{ \Gamma : \sctx \and \vdash A}
              {\sconfig \Gamma A ~\mathit{sort}}\and
              
    \inferrule{ H : \text{heap} \and \Gamma \vdash e : A}
              {\langle H,e\rangle : \sconfig \Gamma A}\and
              
              \inferrule{\Gamma \vdash e : \tnat}{\Gamma \vdash \eget e : \tnat} \and
    
    \inferrule{\Gamma \vdash E : \tnat \rightsquigarrow A\and n : {\mathbb N}}{
      \langle H, E[\eget {(\eret (\natv n))}] \rangle = \langle H, E[H(n)]\rangle : \sconfig \Gamma A}\and
              
    \inferrule{\Gamma \vdash e : \tnat \and
      \Gamma \vdash e' : \tnat
    }{
      \Gamma \vdash \eset e {e'} : \tunit} \and
      
    \inferrule{\Gamma \vdash E : \tunit \rightsquigarrow A \and n : {\mathbb N} \and m : {\mathbb N}}{
      \langle H, E[\eset {(\eret (\natv n))} {\eret (\natv m)}] \rangle = \langle H[n \mapsto m], E[\langle \rangle]\rangle: \sconfig \Gamma A}
  \end{mathpar}
  \caption{STLC state extension}
  \label{fig:state}
\end{figure}

One of our more interesting extensions adds a global mutable heap with set and get expressions,
the rules for which we provide in \autoref{fig:state}.
These rules depend on other extensions not shown, one describing the behavior of global heaps
and another that sets up evaluation contexts.
Heaps are finite maps on natural numbers, defined with the appropriate equations.
Evaluation contexts serve as an interesting example of cross-extension interaction.
When defining pure reductions for features like functions and products,
evaluation contexts are unnecessary since Pyrosome provides congruence.
However, they are essential for the global-heap rules since they allow us
to describe performing a stateful operation inside some larger program.
From a formal perspective, they let us focus on a subterm of the computation
while defining an equation on the whole configuration so that we can access the heap.

It would be more idiomatic in a GAT to formulate such operations with algebraic effects \cite{plotkin2002notions}.
We nonetheless present mutation here using an explicit heap and evaluation contexts,
together with algebraic equations on configurations that define the interactions between programs and heaps,
to demonstrate the breadth of theories that can be encoded as GATs with reasonable naturality.

\autoref{fig:cps-eval-ctx} shows the core constructs of our evaluation-context extension
on the first line, as well as the contexts for STLC as an example.
Plug, which we write $E[e]$, and the empty context $\Ehole$ form the core extension.
We add evaluation contexts for each expression with evaluable subterms
as new pieces of syntax belonging to a sort of evaluation contexts
with judgment form $\Gamma\vdash E : A \rightsquigarrow B$
and define how plug acts on them.
We can then use equational reasoning in our proofs to decompose terms.

As an incidental benefit of our formal, deeply embedded representation of language rules,
we can write a function that generates the full typing rules and equations for evaluation contexts
from short descriptions similar to those that typically appear in a language grammar.
In fact, the four rules mentioning application contexts in \autoref{fig:cps-eval-ctx}
are generated from the following description of the evaluation contexts for STLC:
\newcommand{\Ectxdef}[1]{\{\texttt{#1}\}~}
\[
E ::= ... ~|~  E~ e ~|~ v~ E
\]
The mechanization includes a unique identifier for each evaluation context,
but we overload the syntactic form of the associated expression here, as is common.
The sorts of the subterms are determined by their names.
We denote evaluation contexts with names starting with $E$,
expression subterms with $e$, and value subterms with $v$.
From this description, we generate a well-formedness rule for each context
and the defining equations for the plug operation $\Eplug E {e}$.

\begin{figure}
  \begin{mathpar}
    \inferrule{ }{\Gamma\vdash \Ehole : A \rightsquigarrow A} \and
    \inferrule{\Gamma \vdash E : A \rightsquigarrow B \and \Gamma \vdash e : A}
              {\Gamma \vdash \Eplug E e: B} \and
    \inferrule{\Gamma \vdash e : A}
              {\Gamma \vdash \Eplug \Ehole e = e : A}
  \end{mathpar}
  \begin{mathpar}
    \inferrule{\Gamma \vdash E : A \rightsquigarrow B \to C \and \Gamma \vdash e : B}
              {\Gamma \vdash E~e: A \rightsquigarrow C} \and
    \inferrule{\Gamma \vdash_v v : B\to C \and \Gamma \vdash E : A \rightsquigarrow B}
              {\Gamma \vdash v~E: A \rightsquigarrow C} \and
              \inferrule{\Gamma \vdash E : A \rightsquigarrow B \to C
                \and \Gamma \vdash e : B
                \and \Gamma \vdash e' : A}
              {\Gamma \vdash \Eplug {(E~e)} {e'} = \Eplug E {e'}~e : C}\and
              \inferrule{\Gamma \vdash_v v : B\to C
                \and \Gamma \vdash E : A \rightsquigarrow B
                \and \Gamma \vdash e : A}
              {\Gamma \vdash \Eplug {(v~E)} e = v~\Eplug E {e} : C}
  \end{mathpar}
  \[
  \begin{array}{rcl}
    \cmp{\Ehole} &\triangleq& k~x_h\\
    \cmp{\Eplug E e} &\triangleq&
    \bindCPS x {\cmp e} {\cmp E}\\
    \cmp{E~ e} &\triangleq&
    \bindCPS x {\cmp E}\bindCPS y {\cmp e} {x ~\epair y k}\\
    \cmp{v~E} &\triangleq&
    \bindCPS y {\cmp E} {\cmp v ~\epair y k}\\
      &&\bindCPS x e e' \triangleq {e}[{\elam x {B} {e'}}/k]
      \text{ given } \Gamma,k:\lnot B \vdash e\\
  \end{array}
  \]
  \caption{Evaluation-context extension with STLC}
  \label{fig:cps-eval-ctx}
\end{figure}

Since our CPS translation uses $\bindCPS x e {e'}$ to sequence computations,
we compile away evaluation contexts in our first pass.
The bottom of \autoref{fig:cps-eval-ctx} shows the translation for the core operations
as well as for STLC's evaluation contexts.
This translation maps source-language judgments $\Gamma\vdash E : A \rightsquigarrow B$
to target judgments $\Gamma,k:\lnot B, x_h:A\vdash \cmp{E}$, translating evaluation contexts
into computations with an extra free variable named $x_h$ by convention.
Plug then turns into a bind operation, and holes behave like return operations.

With evaluation contexts established, the translation for our heap operations is minimal.
We can define the CPS pass for the set and get operations as follows:
\[
  \begin{array}{rcl}
    \cmp{\eset e {e'}} &\triangleq&
      \bindCPS x e \!\!\bindCPS y {e'} \!\!\esetCPS x y k~ \langle\rangle \\
    \cmp{\eget e} &\triangleq&
      \bindCPS x e \!\!\egetCPS x y k~y \\
  \end{array}
\]
To match the pattern of the CPS calculus, each heap operation now takes its continuation
as a subterm.
This inversion is why we no longer need to worry about evaluating effects inside of a context.
Since this extension does not interact directly with functions,
we can reuse the CPS version with our closure-converted calculus.
We still have to do a little more than write an identity compiler, however,
since closure conversion changes the expected judgment form, which interacts with the binder in
$\egetCPS x v e$:
\[
  \begin{array}{rcl}
    \cmp{\esetCPS v {v'} e} &\triangleq&
    \esetCPS {\cmp v} {\cmp {v'}} {\cmp e} \\
    \cmp{\egetCPS x v e} &\triangleq&
    \egetCPS x {\cmp v} {e[\epair z x/z]} \\
  \end{array}
  \]

\section{Polymorphism and Advanced Metaprogramming}
\label{sec:poly}

Unlike other case-study extensions, polymorphism necessitates modifying the rules of our simply typed extensions to add the type environment.
To do so, we develop a generic parameterization procedure that takes as input the language to be extended
and a specification of which forms to parameterize, adding the appropriate metavariables.
This procedure is language-agnostic and applies to more than just type environments:
for example we also use it to add type information to our base substitution calculus,
which allows the type substitutions and term substitutions to share a common set of underlying definitions.
We prove that this parameterization function is semantics-preserving on sorts, terms, contexts, and languages.
Using these facts, we further show that it preserves the $\presNoArgs$ predicate on compilers,
allowing the output compilers to participate in all of our modular connections as first-class citizens.

While the high-level idea is simple, formalizing it was complex.
The main challenge of this proof lay in justifying the features that do not need to be parameterized,
for example the sort of natural numbers,
and the position of the additional parameter in the context, when it does not appear as the first component.
Our parameterization procedure takes such language-specific information as inputs,
so the generic proof of parameterization correctness must rely on certain properties of these declarations, like the existence of a dependency ordering such that constructs marked for parameterization cannot be dependencies of ones we do not parameterize.
Fortunately, we reduce them to sound, decidable checks that can guarantee a valid parameterization without inducing additional proof obligations.
The result is \autoref{thm:parameter}.

\begin{theorem}[Parameterization Preserves $\presNoArgs$]
  \label{thm:parameter}
  Let $L_\Delta$, $L_s$, and $L_t$ be well-formed languages.
  Let $s_\Delta$ be a well-formed sort in $L_\Delta$,
  let $\Delta$ be the metavariable name for the new parameter,
  and let $spec_s$ and $spec_t$ be parameterization specifications.
  Let $c$ be a compiler satisfying $\pres{L_t}{c}{L_s}$.
  Then, if all syntactic checks hold on the above components,
  we have
  \[\presExt{id_{L_\Delta}}
  {P_L(\Delta, s_\Delta, spec_t, L_t)}
  {P_c(\Delta, spec_t, spec_s, c)}
  {P_L(\Delta, s_\Delta, spec_s, L_s)}\]
  where $id_{L_\Delta}$ is the identity compiler for $L_\Delta$,
  $P_L$ is the parameterization function for languages,
  and $P_c$ is the parameterization function for compilers.
  
\end{theorem}

\begin{figure}
  \begin{mathpar}
    \inferrule{\Delta,X \vdash A}{\Delta \vdash \forall X. A} \and
    \inferrule{\Delta,X \vdash A}{\Delta \vdash \exists X. A} \and
    \inferrule{\Delta;\Gamma \vdash e : \forall X. A \and \Delta \vdash B}
              {\Delta;\Gamma \vdash e~[B] : A[B/X]} \and
    \inferrule{\Delta\vdash \Gamma \and
      \Delta,X;\Gamma \vdash e : A}{\Delta;\Gamma \vdash_v \Lambda X. e : \forall X. A} \and
    \inferrule{\Delta\vdash \Gamma \and
      \Delta,X;\Gamma \vdash e : A \and \Delta \vdash B}
              {\Delta;\Gamma \vdash (\eret \Lambda X. e)~[B] = e[B/X]}\and
    \inferrule{\Delta;\Gamma \vdash_v v : \forall X. A}
              {\Delta;\Gamma \vdash_v \Lambda X. {v~[X]} = v : \forall X. A}\and
    \inferrule{\Delta;\Gamma \vdash e : \exists X. A\and \Delta\vdash C \and \Delta,X;\Gamma,(x:A) \vdash e' : C}
              {\Delta;\Gamma \vdash \unpack{X}{x}{e}{e'} : C} \and
    \inferrule{\Delta \vdash B \and \Delta;\Gamma \vdash e : A[B/X]}
              {\Delta;\Gamma \vdash \pack{B}{e} : \exists X.A} \and
              
  \end{mathpar}
  \[
  \begin{array}{rcl}
    \cmp{\forall X.A} &\triangleq& \lnot \exists X. \lnot \cmp{A}\\
    \cmp{\Lambda X.e} &\triangleq& \lambda k. \unpack{X}{k}{k}{\cmp{e}}\\
    \cmp{e~[A]} &\triangleq& \bindCPS{x}{\cmp e}{x~ {\pack{\cmp A}{k}}}\\
    \cmp{\exists X.A} &\triangleq&\exists X. \cmp{A}\\
    \cmp{\pack{A}{e}} &\triangleq& \bindCPS{x}{\cmp e}{\pack{\cmp A}{x}}\\
    \cmp{\unpack{X}{x}{e}{e'}} &\triangleq& \bindCPS{y}{\cmp e}{\unpack{X}{x}{y}{\cmp{e'}}}\\
      &&\bindCPS x e e' \triangleq {e}[{\elam x {B} {e'}}/k]
      \text{ given } \Gamma,k:\lnot B \vdash e\\
  \end{array}
  \]
  
  \caption{Polymorphic calculus extensions (selected rules)} 
  \label{fig:system-F}
\end{figure}

Once we have applied the parameterization procedure,
instantiated with the information specific to the type-environment parameter,
to our running case study,
we can develop the type-abstraction and existential-types extensions
to validate that the parameterization process had the intended effect.
We display a portion of the rules for the source-language extension, as well as a selection of CPS-pass cases,
in \autoref{fig:system-F}.
In the target, the parameterized continuation calculus,
we choose to add only existentials, as type abstractions can be compiled to existential continuations via CPS.
Specifically, we translate the type $\forall X.A$ to $\lnot \exists X.\lnot \cmp{A}$.
The closure-conversion extension for existential types is almost an identity,
with the only deviation being some boilerplate in the elimination form to address closure conversion's effect on the environment.



\section{Linear Lambda Calculus}
\label{sec:case-study:linear}
Our primary case study builds entirely on a standard, intuitionistic substitution calculus.
To demonstrate Pyrosome's expressiveness, we also verify a CPS pass analogous to our simply typed one
from a separate, linear lambda calculus to a similarly linear continuation calculus.
The linear substitution calculus is similar to the intuitionistic one, including retaining all rules in \autoref{fig:subst},
but it restricts substitutions and variables to prohibit weakening.
We show the rules for the linear functions extension in \autoref{fig:linear}.
These rules are similar to those of STLC, but in the rules with more than one expression or value,
the context is broken up into multiple (disjoint) pieces.
This separation continues in the substitution rule for application, where the left-hand substitution must break into two parts, one for $e$ and one for $e'$.
Just like with the intuitionistic substitution calculus,
we programmatically generate the last two substitution rules for linear extensions.

The CPS pass itself is actually almost identical to the CPS pass for STLC,
with the only difference being the managing of explicit uses of the exchange rule
as the mechanization is de Bruijn indexed.
This case study currently has more manual proofs since our automation is not yet sufficient to
handle environment splitting without user input, but it still benefits from the metatheory
and principled subgoal management Pyrosome provides.

\begin{figure}
  \begin{mathpar}
    \inferrule{\vdash A \and \vdash B}{\vdash A \multimap B} \and
    \inferrule{\Gamma \vdash e : A \multimap B \and \Gamma' \vdash e' : A}{\Gamma, \Gamma' \vdash e~e' : B} \and
    \inferrule{\Gamma,x:A \vdash e : B}{\Gamma \vdash_v \elam x A e : A \multimap B} \and
    \inferrule{\Gamma,x:A \vdash e : B \and \Gamma' \vdash_v v : A}
              {\Gamma,\Gamma' \vdash (\eret \elam x A e)~(\eret v) = e[v/x] : B}\and
              \inferrule{\Gamma \vdash e : A \multimap B \and \Delta \vdash e' : A
                \and \gamma : \Gamma' \Rightarrow \Gamma
                \and \delta : \Delta' \Rightarrow \Delta}
              {\Gamma',\Delta' \vdash (e~e')[\gamma,\delta] = e[\gamma]~e'[\delta] : B} \and
    \inferrule{\Gamma,x:A \vdash e : B \and \gamma : \Gamma' \Rightarrow \Gamma}
              {\Gamma' \vdash_v (\elam x A e)[\gamma] = \elam x A e[\gamma,x/x]: A \multimap B}
  \end{mathpar}
  \caption{Linear Lambda Calculus}
  \label{fig:linear}
\end{figure}

\section{IMP}
\label{sec:case-study:imp}
To show that our framework is not restricted to functional languages,
we also compile an imperative calculus,
the grammar of which is shown in \autoref{fig:imp-bnf},
similar to the one in Chaper 2 of \citet{Winskel}
to the same target as our functional code.
We reuse the heap definition from our functional heap extension to model memory,
since memory locations without pointer arithmetic are equivalent to global variable names.
Imperative statements translate to target-language computations that expect a continuation with unit argument,
and expressions translate to computations where the continuation must be passed a natural number.
To support this compiler, we add a conditional to the target language that branches based on whether the input value is 0.
Since the target-language form only accepts values (including variables) in the condition,
the compilation for source-level conditionals first evaluates the condition
and then branches on the result.
The compilation for while loops works similarly, with the added complication of a recursion construct
so that it can jump back to the condition at the end of an iteration.

Note that since this compiler shares a common target with our main case study,
we obtain some guarantees about the interoperation of imperative and functional code.
Specifically, they can be linked, directly or with target-level glue code,
so long as the target types line up,
and we can prove that linking a compiled imperative program $\cmp S$ with a compiled functional program $\cmp f$ by either method
produces an equivalent program to linking some $\cmp {S'}$ and $\cmp{f'}$ the same way
if $S = S'$ and $f = f'$ in their respective source languages.

\begin{figure}
\[
\begin{array}{rcl}
    a,a' &::=& n ~|~ x ~|~ a + a' ~|~ a - a'\\
    s,s' &::=& \Iskip ~|~ \Iassign x a ~|~ \Iseq{s}{s'} ~|~ \ite a s {s'} ~|~ \Iwhile a s \\
    n,x &\in& \mathbb{N}\\
    \\
\end{array}
\]
  \[
  \begin{array}{rcl}
    \cmp{\Iskip} &\triangleq& k~\langle\rangle\\
    \cmp{\Iassign x a} &\triangleq&
    \bindCPS y {\cmp a}{\esetCPS {\cmp x} y {k~\langle\rangle}}\\
    \cmp{s;s'} &\triangleq&
    \bindCPS \_ {\cmp s} {\cmp {s'}}\\
    \cmp{\ite a s {s'}} &\triangleq&
    \bindCPS y {\cmp a}{\ite y {\cmp s} {\cmp {s'}}}\\
    \cmp{\Iwhile b s} &\triangleq&
           {\efixCC {\text{clo }\left\langle
               \begin{array}{l}
                 z : {\neg \tunit} \times {({\neg \tunit} \times \tunit)}.\\
                 \bindCPS x {\cmp b}
                   \\
              {\ite x
                {{z.1}~ {\langle\rangle}}
                
                {{\cmp s}[z.2.1/k]}

            }
             \end{array}
            ,
              k\right\rangle}}
           ~{\langle\rangle}

  \end{array}
  \]
  \caption{IMP (excerpt)}
  \label{fig:imp-bnf}
\end{figure}

\section{Inference and Automation}
\label{sec:misc:infer}
We rely on automation for type inference and to do most of the heavy lifting in our proofs.
Formally, GATs require the user to specify values for every metavariable when using a term former.
For example, an application node must include the input and output types, as well as the context,
as additional arguments.
To make writing GATs in our mechanization practical, we use the term below the line to determine which arguments must be written and which should be inferred.
To return to application, since we write it $e~e'$, the arguments $e$ and $e'$ are explicit
while $\Gamma$, $A$, and $B$ are implicit.

To support inferring implicit arguments, we define a separate elaboration judgment in parallel with each well-formedness judgment of Pyrosome, including term well-formedness, language well-formedness, and $\presNoArgs$.
These versions take in an extra parameter, the pre-elaboration syntax,
and ensure that the term, language, or compiler under scrutiny lines up with it.
We then derive the full term, language, or compiler
with tactics that generate correct-by-construction elaborations.
The tactics we developed are fairly general and bear no ties to any of the languages in our case study, including to the substitution calculus, as shown by the fact that the same inference tactics also work on the linear and imperative examples.

\begin{wrapfigure}{R}{0.55\textwidth}
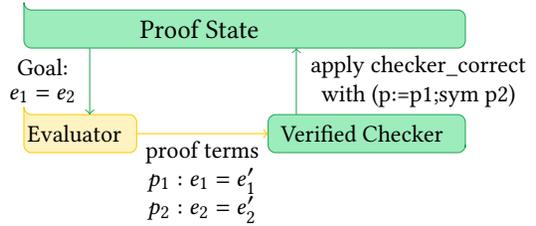

  \vspace{-0.1in}
\ctikzfig{automation}
\vspace{-12pt}
\caption{Architecture of equation-solving tactic}
\label{fig:automation-diagram}
\vspace{-0.1in}
\end{wrapfigure}

We largely automate proofs of equivalence preservation using a tactic that normalizes both sides of the goal,
the structure of which is shown in \autoref{fig:automation-diagram}.
To do so, we rely on the convention that the equations specified in each language can be read left-to-right as reduction rules.
This convention is purely a heuristic convenience, but on the examples we have worked through it has been highly successful.  
Since both our terms and language definitions are represented as deeply embedded syntax,
we define a Gallina function, the unverified evaluator in \autoref{fig:automation-diagram},
that takes as input a term and the language it belongs to
and rewrites it from left to right along each applicable rule.
To be useful in our proofs, this evaluator generates a reified proof term representing the equivalence between its input and output.
We additionally define and verify a proof-checking function that computes whether such a proof is valid.
Our tactic then runs the evaluator on the terms in the goal and inserts a call to the checker with the result in a proof-by-reflection style.
Structuring the tactic this way means that we do not have to verify the evaluator code, which is notably more complex than the proof checker. As an added benefit, since the evaluator does not appear in the proof term, only the checking function must be rerun at the end of the proof, which improves performance.
This approach turned out to be effective, solving almost all of the equivalence-preservation goals in our case study.
The obligations that required manual guiding generally fit at least one of two categories: either they required the use of an $\eta$-law, or they interacted with the heap.
Our tactics featured reduced efficacy on heap rules related to the encoding of pointer comparisons in our rule definitions.

\section{Related Work}
\label{sec:related}

\subsection{Alternative Generic Frameworks}
\label{sec:related:alt}
The principle benefit of GATs \cite{cartmell}
over most generic frameworks for programming-language metatheory \cite{felleisen1991expressive,MTC,3MT}
is its presentation via equational theories rather than operational semantics.
\citet{felleisen1991expressive} gives meaning to its programs
by an arbitrary termination predicate and a contextual-equivalence relation
defined in terms of that predicate.
As discussed in \autoref{sec:ext:ctx-eq}, contextual equivalence is too strong for extensible reasoning.
We, by design, choose not to include every reasonable equation in a given language,
since equations that may be reasonable to include, say, in a pure language
interact poorly with effectful extensions.
By using the smallest equational theories wherever possible,
that is only including the equations at each level that application and compiler reasoning demand,
users of Pyrosome increase the compatibility and extensibility of their developments.

Prior work on modular metatheory and language extensibility
mechanized a system in Coq that covered components like language interpreters,
including modular soundness properties \cite{MTC,3MT}, in the style of \citet{datatypes-a-la-carte}.
\citet{Allais} also built a framework for reasoning about binding, renaming, and translation
that bears some similarities to this work.
There exist many other formal frameworks in the same general spirit, including a long line of work
on logical frameworks \cite{harper1993framework}, that are quite useful in broad domains,
but assume a particular structure for variable binding.
However, these projects incorporate object-language binding and substitution into their frameworks,
rather than implementing them as a first-class feature.
While this is attractive for ease of use since it lifts substitution into the metatheory,
it restricts the generality of the framework in ways that we wish to avoid in Pyrosome.

The K Framework \cite{rosu2010overview} has built an extensive ecosystem
of generic language-specification tooling,
and their logic is powerful enough to express a variety of binders internally \cite{k-binders}.
The project has also recently expanded to cover certified proofs of the behavior of individual programs \cite{k-proof}.
However, they do not address the higher-order concerns of verified compilation
and compiler extensibility that we cover in this work.

The recent GATlab project \cite{Lynch_2024} is based on a very similar system of GATs and GAT morphisms,
implemented as an embedded domain-specific language in Juila for the purpose of algebraic modelling in scientific and engineering applications.
GATLab offers a number of methods of interaction designed to facilitate exploration in applied category theory, which would likely be useful additions to Pyrosome for the purposes of GAT development.
However, as it is not a theorem prover, it does not attempt to prove facts, internally or metatheoretically, as we do.

\subsection{Multilanguage Semantics}
\label{sec:multilang}

Existing literature documents some of the uses of multilanguage semantics,
including in compiler verification \cite{multilang,perconti2014verifying,ML-verif-comp}.
However, rather than use a formal framework to describe the way these works combine languages,
the authors design ad-hoc multilanguages for their specific use cases.
As a result, they cannot exploit the generic properties of language extension
that we take advantage of.
Furthermore, since they use contextual equivalence in their multilanguages
in their specifications of compiler correctness \cite{perconti2014verifying},
they cannot extend their compilers with new supported features or additional passes.
Related work takes steps towards generalizing the key elements of this line of research
so as to better support modular reasoning and extensibility \cite{bowmancompilation}
by expressing compiler cases as operational rules,
but it is still limited by its use of contextual equivalence.
While contextual equivalence results do provide certain guarantees with regard to security that we do not consider,
we believe that a more compositional paradigm is necessary to reduce the burden of multilanguage reasoning.

Recent work on multilanguage semantics enables complex cross-language interoperation
by directly defining the semantics of the source languages in terms of a target-language logical relation \cite{patterson-semantics}. This enables a rich variety of FFI designs and multilanguage interoperation. However, since the source languages do not have their own semantics, it gives up source-level reasoning, which reduces the benefit of writing higher-level code, especially in verified settings. This approach uses a relation that generalizes the compiler as the specification of source-language behavior, so it does not adapt well to our primary concern of compiler correctness. In contrast, we define a source-level equational theory, which allows us to separate the specification of each feature's behavior from the compiler's implementation, verify that the former describes the latter, and thereafter reason only at the source level.

\subsection{Existing Verified Compilers}

Since Pyrosome currently is a foundational study rather than a production tool,
we cannot compare it directly to the tooling built up for more mature projects like CakeML \cite{cake} and CompCert \cite{compcert,compcerto}.
However, we can discuss what improvements or extensions to our theoretical results might be necessary before compilers with similar feature sets can be implemented in Pyrosome.
One major hurdle is developing our story for optimization as discussed in \autoref{sec:optimization}.
While prior work shows that real optimizations can be verified against equational theories,
it remains to properly explore the design space regarding how best to make them extensible.

The other significant gap between what we theoretically support and existing verified-compiler ecosystems is our lack of refinement relations since GATs by definition define language semantics in terms of equality.
There are a number of ways to encode refinement or nondeterministic computation,
ranging from defining the relation internally in the object language
to viewing nondeterministic operations as computing sets of results.
However, the most straightforward approach to incorporating refinement would be to generalize the framework itself.
None of the core metatheory in Pyrosome depends on symmetry, so it should be possible to remove it as an axiom
and explore the expanded space of theories.
We would also need to consider how to handle translating programs from Turing-complete languages like System F to finite state machines like assembly languages, possibly by using refinement and an approach like \citet{beck2024two}.
CompCert also features graph-shaped IRs, which we have yet to investigate.

\subsection{DimSum}
The recent DimSum framework \cite{dimsum} gives a convincing accounting of language-agnostic compiler verification and cross-language linking for low-level, event-based systems.
However, we see Pyrosome and DimSum as complementary.
DimSum's design addresses cross-language library interaction at preexisting
language interfaces, primarily procedure calls in their examples.
However, it does not provide machinery for extending the languages themselves,
especially with new features that do not fit the preexisting event type.
Such reasoning might be possible via some clever design pattern that expresses each language feature as an event,
but this approach is clearly not an idiomatic or intended way to use their system.

Unlike DimSum, Pyrosome's syntactic approach based on equational theories allows us to extend a language's syntax and semantics at a fine-grained level. 
In addition to deriving results about the interactions of programs across languages,
we enable users of Pyrosome to reason about programs, languages, and compilers
in a way that is preserved by language extension as well.
Our syntactic core also supports both programmatic manipulation of language definitions
and our extensive automation of proof goals.

In reference to prior work on multilanguage semantics \cite{multilang,perconti2014verifying,ML-verif-comp},
\citet{dimsum} conclude that ``syntactic multi-languages scale
well to typed, higher-order languages. In [the DimSum] paper, we have put the focus on different kinds of
languages: untyped, low-level languages comparable to C and assembly.''
We largely agree with this analysis.
As we discuss above in \ref{sec:multilang} and as the DimSum authors concur,
prior work on constructing multilanguages suffered from difficulties with
the ad-hoc nature of multilanguage construction and the limitations of contextual equivalence.
Since we have addressed both of these concerns in Pyrosome, we see the methodology of Pyrosome
as the preferred approach in typed and/or higher-order settings.
In particular, while Pyrosome allows linking and compiler-correctness results
to depend on the guarantees of the involved languages' type systems,
it is unclear how to express such properties in DimSum without a fundamental extension of their framework.

\section{Future Work}

As-is, compilation theorems proven using \fwk{} include the initial source-language
and final target-language specifications in their trusted base, in addition to Pyrosome's core definitions.
We hope to bridge the gap between existing formal semantics and ones
defined in \fwk{} so that end users can benefit from the extensibility
of engineering done in \fwk{} without having to trust \fwk{}'s definitions.
The most important step in this process is to develop models of target-language theories
implemented by verified low-level systems.
Although we currently only target \fwk{} languages with our case studies,
our compilers and metatheorems support any target model that validates core well-formedness and equivalence properties.
We hope in the future to target established projects like CompCert~\cite{compcert} or Bedrock~\cite{lightbulb}
so that upper levels of a compiler can use the flexibility of Pyrosome
while benefiting from established codebases beneath.
Target models could even be denotational in nature.
For example, work using interaction trees already tends toward algebraic reasoning via corpuses of equational lemmas \cite{xia2019interaction}.
It may be worth investigating interaction trees as a denotation for Pyrosome languages,
since the combination could leverage Pyrosome's generic tooling to prove facts about interaction trees.

Additionally, \fwk{} was designed to support a wide range of features at both the term and type levels.
Originally developed for dependent types, the theory behind Pyrosome has been shown to support a variety of interesting type-level features \cite{gats,cartmell} in language definitions.
For example, we have mechanized a dependent variant of our substitution language,
in other words Martin-L\"of Type Theory without connectives,
following the descriptions in prior work \cite{dybjer1995internal}.
We expect a sufficiently motivated user could extend the formalization of this calculus
with the usual dependent connectives, such as dependent products and sums, or even a hierarchy of universes \cite{gats}.
Such extensions would rely on our monotonicity and compiler-extension theorems just like our case study.
However, in this paper we chose to limit ourselves to features that fit within our current proof automation.
At present many dependent connectives would require more manual proofs due to
our current automation's reliance on syntactic type equality during type inference.
This restriction is not related to our theorems but is rather an issue of tactic-engineering complexity,
and we hope to extend our automation to support such features better in the future.

\begin{acks}                            
  We thank Amal Ahmed for early guidance on this project and the anonymous reviewers for extensive constructive feedback.
  This research was supported by the \grantsponsor{GS13}{National Science Foundation Graduate Research Fellowship Program}{} under Grant No. \grantnum{GS13}{2141064}.
  Any opinions, findings, and conclusions or recommendations expressed in this material are those of the author(s) and do not necessarily reflect the views of the National Science Foundation.
\end{acks}

\bibliography{bibliography.bib}

\end{document}